\documentclass{article}

\usepackage{amsmath}
\usepackage{amsthm}
\usepackage{amssymb}
\usepackage{tikz}
\usepackage{url}
\usepackage{authblk}

\usetikzlibrary{arrows,automata}

\newcommand{\np}{\mathsf{NP}}
\newcommand{\T}{\mathcal{T}}
\renewcommand{\S}{\mathcal{S}}
\newcommand{\GM}{\textsc{Graph Motif}}
\newcommand{\pcgm}{\textsc{CGM}}
\newcommand{\pxcgm}{\textsc{XCGM}}
\newcommand{\pmgm}{\textsc{MGM}}
\newcommand{\pmgmg}{\textsc{MGMG}}
\newcommand{\pxmgm}{\textsc{XMGM}}
\newcommand{\pwcgm}{\textsc{WCGM}}
\newcommand{\pwmgm}{\textsc{WMGM}}
\newcommand{\pmld}{\textsc{MLD}}
\newcommand{\pxmld}{\textsc{XMLD}}
\newcommand{\pmest}{\textsc{MEST}}

\newcommand{\wone}{\mathsf{W[1]}}
\newcommand{\swone}{\mathsf{\# W[1]}}
\newcommand{\nat}{\mathbb{N}}
\newcommand{\Z}{\mathbb{Z}}
\renewcommand{\P}{\mathcal{P}}
\newcommand{\U}{\mathcal{U}}
\newcommand{\K}{\mathcal{K}}
\newcommand{\A}{\mathbb{A}}
\newcommand{\C}{\mathcal{C}}
\newcommand{\bigo}{\mathcal{O}}
\newcommand{\bigos}{\mathcal{O}^*}
\newcommand{\bigot}{\tilde{\mathcal{O}}}
\newcommand{\lmul}{\{ \hspace{-0.1cm} \{}
\newcommand{\rmul}{\} \hspace{-0.1cm} \}}

\newtheorem{theorem}{Theorem}
\newtheorem{proposition}{Proposition}
\newtheorem{lemma}{Lemma}

\newcommand{\prog}[1]{\textsf{#1}}

\begin{document}

\title{Finding and Counting Vertex-Colored Subtrees \thanks{An extended abstract of this paper appeared in MFCS 2010 \cite{guillemot:sikora:mfcs:2010}.}}

\author[1,2]{Sylvain Guillemot \thanks{\texttt{sguillem@iastate.edu}}}
\author[2,3]{Florian Sikora \thanks{\texttt{sikora@univ-mlv.fr}}}

\affil[1]{Department of Computer Science, Iowa State University, Ames, IA 50011, USA.}
\affil[2]{Lehrstuhl f\"ur Bioinformatik, Friedrich-Schiller Universit\"at Jena, Ernst-Abbe Platz~2, 00743 Jena, Germany}
\affil[3]{Universit\'{e} Paris-Est, LIGM - UMR CNRS 8049, 77454 Marne-la-Vallée Cedex 2, France}


\date{}


\maketitle

\begin{abstract}
The problems studied in this article originate from the \sloppy $\GM$ problem introduced by Lacroix et al. \cite{Lacroix:Fernandes:Sagot:TCBB:2006} in the context of biological networks. The problem is to decide if a vertex-colored graph has a connected subgraph whose colors equal a given multiset of colors $M$. It is a graph pattern-matching problem variant, where the structure of the occurrence of the pattern is not of interest but the only requirement is the connectedness. Using an algebraic framework recently introduced by Koutis et al. \cite{Koutis:icalp:2008,Koutis:Williams:ICALP:2009}, we obtain new FPT algorithms for $\GM$ and variants, with improved running times. We also obtain results on the counting versions of this problem, proving that the counting problem is FPT if $M$ is a set, but becomes $\swone$-hard if $M$ is a multiset with two colors. Finally, we present an experimental evaluation of this approach on real datasets, showing that its performance compares favorably with existing software.

\end{abstract}

\section{Introduction}
An emerging field in the modern biology is the study of the biological networks, which represent the interactions between biological elements \cite{Alm:Arkin:2003}. A network is modeled by a vertex-colored graph, where nodes represent the biological compounds, edges represent their interactions, and colors represent functionalities of the graph nodes. Networks are often analyzed by studying their \emph{network motifs}, which are defined as small recurring subnetworks. Motifs generally correspond to a set of elements realizing a same function, and which may have been evolutionarily preserved. Therefore, the discovery and the querying of motifs is a crucial problem \cite{Sharan:Ideker:2006}, since it can help to decompose the network into functional modules, to identify conserved elements, and to transfer biological knowledge across species.

The initial definition of network motifs involves conservation of the topology and of the node labels; hence, looking for topological motifs is roughly equivalent to subgraph isomorphism, and thus is a computationally difficult problem. However, in some situations, the topology is not known or is irrelevant, which leads to searching for \emph{functional} motifs instead of \emph{topological} ones. In this setting, we still ask for the conservation of the node labels, but we replace topology conservation by the weaker requirement that the subnetwork should form a connected subgraph of the target graph. This approach was advocated by \cite{Lacroix:Fernandes:Sagot:TCBB:2006} and led to the definition of the $\GM$ problem \cite{Fellows:Fertin:Hermelin:Vialette:ICALP:2007}: given a vertex-colored graph $G = (V,E)$ and a multiset of colors $M$, find a set $V' \subseteq V$ such that the induced subgraph $G[V']$ is connected, and the multiset of colors of the vertices of $V'$ is equal to $M$. In the literature, a distinction is made between the \emph{colorful} case (when $M$ is a set), and the \emph{multiset} case (when $M$ is an arbitrary multiset). Although this problem has been introduced for biological motivations, \cite{Betzler:Fellows:Komusiewicz:Niedermeier:CPM:2008} points out that it may also be used in social or technical networks.

Not surprisingly, $\GM$ is $\np$-hard, even if $G$ is a bipartite graph with maximum degree 4 and $M$ is built over two colors only \cite{Fellows:Fertin:Hermelin:Vialette:ICALP:2007}. The problem is still $\np$-hard if $G$ is a tree of diameter four and $M$ is a set \cite{Ambalath:IPEC:2010}. However, for general trees and multiset motifs, it can be solved in $\bigo(n^{2c+2})$ time, where $c$ is the number of distinct colors in $M$, while being $\wone$-hard for the parameter $c$ \cite{Fellows:Fertin:Hermelin:Vialette:ICALP:2007}. The difficulty of the problem is counterbalanced by its fixed-parameter tractability when the parameter is $k$, the size of the solution \cite{Lacroix:Fernandes:Sagot:TCBB:2006,Fellows:Fertin:Hermelin:Vialette:ICALP:2007,Betzler:Fellows:Komusiewicz:Niedermeier:CPM:2008}. The currently fastest FPT algorithms for the problem run in $\bigos(2^k)$ time for the colorful case, $\bigos(4.32^k)$ time for the multiset case, and use exponential space. Throughout the paper, we use the notations $\bigos$ and $\bigot$ to suppress polynomial and polylogarithmic factors, respectively. In addition to these results, it was shown in~\cite{Ambalath:IPEC:2010} that the problem is unlikely to admit polynomial kernels, even on restricted classes of graphs.

Our contribution is twofold. First, we consider in Section \ref{sec:fpt} the \emph{decision} versions of the $\GM$ problem, as well as some variants: we obtain improved FPT algorithms for these problems, by using the algebraic framework of \emph{multilinear detection} for arithmetic circuits \cite{Koutis:icalp:2008,Koutis:Williams:ICALP:2009}, presented in the next section. Second, we investigate in Section \ref{sec:counting} the \emph{counting} versions of the $\GM$ problem: instead of deciding if a motif appears in the graph, we now want to count the occurrences of this motif. This allows to assess if a motif is over- or under- represented in the network, by comparing the actual count of the motif to its expected count under a null hypothesis \cite{SLS09}. We show that the counting problem is FPT in the colorful case, but becomes $\swone$-hard for the multiset case with two colors. We refer the reader to \cite{Flum:Grohe:2006,Flum:Grohe:SIAM:2004} for definitions related to parameterized counting classes.

\section{Definitions}

This section contains definitions related to arithmetic circuits, and to the \textsc{Multilinear Detection} ($\pmld$) problem. It concludes by stating Theorem \ref{th:mld}, which will be used in Section~\ref{sec:fpt}.

\subsection{Arithmetic circuits}

In the following, a capital letter $X$ will denote a set of variables, and a lowercase letter $x$ will denote a single variable.
If $X$ is a set of variables and $\A$ is a commutative ring, we denote by $\A[X]$ the ring of multivariate polynomials with coefficients in $\A$ and involving variables of $X$. Given a monomial $m = x_1 \dots x_k$ in $\A[X]$, where the $x_i$s are variables, its \emph{degree} is $k$, and $m$ is \emph{multilinear} iff its variables are distinct.

 An \emph{arithmetic circuit} over $X$ is a pair $\C = (C,r)$, where $C$ is a labeled directed acyclic graph (dag) such that (i) the children of each node are totally ordered, (ii) the nodes are labeled either by $op \in \{+,\times\}$ or by an element of $X$, (iii) no internal node is labeled by an element of $X$, and where $r$ is a distinguished node of $C$ called the \emph{root} (see Figure \ref{fig:circuit} for an example of $C$). We denote by $V_{\C}$ the set of nodes of $C$, and for a given node $u$ we denote by $N_C(u)$ the set of children (\textit{i.e.} out-neighbors) of $u$ in $C$. We recall that a node $u$ is called a \emph{leaf} of $C$ iff $N_C(u) = \emptyset$, an \emph{internal node} otherwise. We denote by $T(\C)$ the \emph{size of $\C$} (defined as the number of arcs), and we denote by $S(\C)$ the number of nodes of $\C$ of indegree $\geq 2$.

Given a commutative ring $\A$, \emph{evaluating $\C$ over $\A$ under a mapping $\phi : X \rightarrow \A$} consists in computing, for each node $u$ of $C$, a value $val(u) \in \A$ as follows: 1. for a leaf $u$ labeled by $x \in X$, we let $val(u) = \phi(x)$, 2. for an node $u$ labeled by $+$ (resp. $\times$), we compute $val(u)$ as the sum (resp. product) of the values of its children. The result of the evaluation is then $val(r)$. 
By convention, an empty sum evaluates to $0_{\mathbb{A}}$, and an empty product evaluates to $1_\mathbb{A}$ ($\mathbb{A}$ is assumed to be a unital ring).
The \emph{symbolic evaluation} of $\C$ is the polynomial $P_{\C} \in \Z[X]$ obtained by evaluating $\C$ over $\Z[X]$ under the identity mapping $\phi : X \rightarrow \Z[X]$. We stress that the above definition of arithmetic circuits does not allow constants, a restriction which is necessary for the algorithms. 


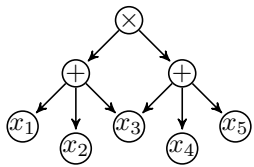
\begin{figure}[tb]
  \begin{center}
\begin{tikzpicture}[->,>=stealth',shorten >=1pt,auto,semithick]
  \tikzstyle{every state}=[circle,fill=black!0,minimum size=9pt,inner sep=0pt]
  \tikzstyle{vertex}=[circle,fill=black!0,minimum size=9pt,inner sep=0pt]

  \node[state] 		   (add1)                    {$+$};
  \node[state]         (mult1) [above right of=add1] {$\times$};
  \node[state]         (add2) [below right of=mult1] {$+$};
  \node[state]         (x1) [below left of=add1] {$x_1$};
  \node[state]         (x2) [below of=add1] {$x_2$};
  \node[state]         (x3) [below right of=add1] {$x_3$};
  \node[state]         (x4) [below of=add2] {$x_4$};
  \node[state]         (x5) [below right of=add2] {$x_5$};

  \path 
        (mult1) edge              (add2)
            	edge              (add1)
		(add1) edge (x1)
				edge (x2)
				edge (x3)
		(add2) edge (x3)
		(add2) edge (x4)
		(add2) edge (x5)            	
            	;
	
	 \end{tikzpicture}
  \end{center}
  \caption{The labeled dag representing the polynomial $(x_1+x_2+x_3)(x_3+x_4+x_5)$.}\label{fig:circuit}
\end{figure}

\subsection{Multilinear Detection}

Informally, the \textsc{Multilinear Detection} problem asks, for a given arithmetic circuit $\C$ and an integer $k$, if the polynomial $P_{\C}$ has a multilinear monomial of degree $k$. However, this definition does not give a certificate checkable in polynomial-time, so for technical reasons we define the problem differently.

A \emph{monomial-subtree} of $\C$ is a pair $T = (\C',\phi)$, where $\C' = (C',r')$ is an arithmetic circuit over $X$ whose underlying dag $C'$ is a directed tree, and where $\phi : V_{\C'} \rightarrow V_{\C}$ is such that (i) $\phi(r') = r$, (ii) if $u \in V_{\C'}$ is labeled by $x \in X$, then so is $\phi(u)$, (iii) if $u \in V_{\C'}$ is labeled by $+$ then so is $\phi(u)$, and $N_{\C'}(u)$ consists of a single element $v \in N_{\C}(\phi(u))$, (iv) if $u \in V_{\C'}$ is labeled by $\times$, then so is $\phi(u)$, and $\phi$ maps bijectively $N_{\C'}(u)$ into $N_{\C}(\phi(u))$ by preserving the ordering on siblings. 
By the \textit{variables} of $T$, we mean the variables of $X$ labelling the leaves of $C'$. 
We say that $T$ is \emph{distinctly-labeled} iff its variables are distinct.

Intuitively, a monomial-subtree tells us how to construct a monomial from the circuit: Condition (i) tells us to start at the root, Condition (iii) tells us that when reaching a $+$ node we are only allowed to pick one child, and Condition (iv) tells us that when reaching a $\times$ node we have to pick all children. The (distinctly-labeled) monomial-subtrees of $\C$ with $k$ variables will then correspond to the (multilinear) monomials of $P_{\C}$ having degree $k$. Therefore, we formulate the \textsc{Multilinear Detection} problem as follows.\\
~

\noindent \textbf{Name:} \textsc{Multilinear Detection} ($\pmld$)\\
\textbf{Input:} An arithmetic circuit $\C$ over a set of variables $X$, an integer $k$.\\
\textbf{Solution:} A distinctly-labeled monomial-subtree of $\C$ with $k$ variables.\\

Solving $\pmld$ amounts to decide if $P_{\C}$ has a multilinear monomial of degree $k$, 
and solving $\# \pmld$ amounts to compute the sum of the coefficients of multilinear monomials of $P_{\C}$ having degree $k$.
The restriction of $\pmld$ when $|X| = k$ is called \textsc{Exact Multilinear Detection} ($\pxmld$). In this article, we will rely on the following far-reaching result from \cite{Williams:2009,Koutis:Williams:ICALP:2009} to obtain new algorithms for \textsc{Graph Motif}:

\begin{theorem}[\cite{Williams:2009,Koutis:Williams:ICALP:2009}]\label{th:mld} $\pmld$ can be solved by a randomized algorithm which uses $\bigot(2^k T(\C))$ time and $\bigot(S(\C))$ space.
\end{theorem}


\section{Finding vertex-colored subtrees\label{sec:fpt}}

In this section, we consider several variants of the $\GM$ problem, and we obtain improved FPT algorithms for these problems by reduction to $\pmld$. Notably, we obtain $\bigos(2^k)$ time algorithms for problems involving \emph{colorful} motifs, and $\bigos(4^k)$ time algorithms for \emph{multiset} motifs.

\subsection{The colorful case}

In the colorful formulation of the problem, the graph is vertex-colored, and we seek a subtree with $k$ vertices having distinct colors. This leads to the following formal definition.\\
~

\noindent \textbf{Name:} \textsc{Colorful Graph Motif} ($\pcgm$)\\
\textbf{Input:} A graph $G = (V,E)$, a set $C$, a function $\chi : V \rightarrow C$, an integer $k$.\\
\textbf{Solution:} A subtree $T = (V_T,E_T)$ of $G$ s.t. (i) $|V_T| = k$ and (ii) for each $u,v \in V_T$ distinct, $\chi(u) \neq \chi(v)$.\\
~

The restriction of \textsc{Colorful Graph Motif} when $|C| = k$ is called \textsc{Exact Colorful Graph Motif} ($\pxcgm$). Note that this restriction requires that the vertices of $T$ are bijectively labeled by the colors of $C$. In \cite{Bruckner:2009}, the $\pxcgm$ problem was shown to be solvable in $\bigos(2^k)$ time and space, while it is not difficult to see that the general $\pcgm$ problem can be solved in $\bigos((2e)^k)$ time and $\bigos(2^k)$ space by color-coding. By using a reduction to \textsc{Multilinear Detection}, we improve upon these complexities. In the following, we let $n$ and $m$ denote the number of vertices and the number of edges of $G$, respectively.

\begin{proposition} \label{prop:cgm} $\pcgm$ is solvable by a randomized algorithm in $\bigot(2^k k^2 m)$ time and $\bigot(k n)$ space.
\end{proposition}

\begin{proof} Let $I$ be an instance of $\pcgm$. We construct the following circuit $\C_I$: its set of variables is $\{ x_c : c \in C \}$, and we introduce intermediate nodes $P_{i,u}$ for $1 \leq i \leq k, u \in V$, as well as a root node $P$. Informally, the multilinear monomials of $P_{i,u}$ will correspond to distinctly colored subtrees of $G$ having $i$ vertices, including $u$. The definitions are as follows:
\begin{align*}
 P_{1,u} &= x_{\chi(u)}\\
 P_{i,u} &= \sum_{i'=1}^{i-1} \sum_{v \in N_G(u)} P_{i',u} P_{i-i',v} \text{ if $i > 1$}
\end{align*}
and $P = \sum_{u \in V} P_{k,u}$. The resulting instance of $\pmld$ is $I' = (\C_I,k)$. The number of arcs in the circuit is $T(\C_I) = \bigo(k^2 m)$ and the number of nodes with indegree $\geq 2$ is $S(\C_I) = \bigo(k n)$. Indeed, $P$ has $n$ children, each $P_{1,u}$ is a leaf and for each $i>1$, each $P_{i,u}$ creates $(i-1) \cdot (3deg(u) + 1) \leq 3k \cdot deg(u)$ arcs (where $deg(u)$ is the degree of $u$, which is assumed w.l.o.g. to be strictly positive). Therefore, the number of arcs in $\C_I$ is $n + \sum_{i=2}^k \sum_{u \in V} 3k \cdot deg(u) \leq n + 6k^2 m$. The number of nodes of $\C_I$ of indegree $\geq 2$ is straightforward since there is at most $kn$ different nodes $P_{i,u}$. Consequently, by applying Theorem \ref{th:mld}, we solve $I'$ in $\bigot(2^k k^2 m)$ time and $\bigot(k n)$ space.


It remains to show the correctness of the reduction. Given a set $S \subseteq C$, define the multilinear monomial $\pi_S := \prod_{c \in S} x_c$. Given $u \in V$ and $S \subseteq C$, an \emph{$(u,S)$-solution} is a subtree $T = (V_T,E_T)$ of $G$, such that $u \in V_T$, $T$ is distinctly colored by $\chi$, and $\chi(V_T) = S$. We show by induction on $1 \leq i \leq k$ that: $\pi_S$ is a multilinear monomial of $P_{i,u}$ iff (i) $|S| = i$ and (ii) there exists an $(u,S)$-solution. This is clear when $i = 1$; now, suppose that $i \geq 2$, and assume that the property holds for every $1 \leq j < i$.

Suppose that $|S| = i$ and that $T = (V_T,E_T)$ is an $(u,S)$-solution, let us show that $\pi_S$ is a multilinear monomial of $P_{i,u}$. Let $v$ be a neighbor of $u$ in $T$, then removing the edge $uv$ from $T$ produces two trees $T_1,T_2$ with $T_1$ containing $u$ and $T_2$ containing $v$. These
two trees are distinctly colored, let $S_1,S_2$ be their respective color sets, and let $i_1,i_2$ be
their respective sizes. Since $T_1$ is an $(u,S_1)$-solution, $\pi_{S_1}$ is a multilinear monomial of $P_{i_1,u}$ by the induction hypothesis. Since $T_2$ is a $(v,S_2)$-solution, $\pi_{S_2}$ is a multilinear monomial of $P_{i_2,v}$ by the induction hypothesis. It follows that $\pi_S = \pi_{S_1} \pi_{S_2}$ is a multilinear monomial of $P_{i_1,u} P_{i_2,v}$, and thus of $P_{i,u}$.

Conversely, suppose that $\pi_S$ is a multilinear monomial of $P_{i,u}$. By definition of $P_{i,u}$, there exists $1 \leq i' \leq i-1$ and $v \in N_G(u)$ such that $\pi_S$ is a multilinear monomial of $P_{i',u} P_{i-i',v}$. We can then partition $S$ into $S_1,S_2$, with $\pi_{S_1}$ multilinear monomial of $P_{i',u}$
and $\pi_{S_2}$ multilinear monomial of $P_{i-i',v}$. The induction hypothesis therefore implies that (i) $|S_1| = i'$ and $|S_2|�= i-i'$, (ii) there exists an $(u,S_1)$-solution $T_1 = (V_1,E_1)$ and a $(v,S_2)$-solution $T_2 = (V_2,E_2)$. Since $S_1,S_2$ are disjoint, it follows that $|S| = i$, which proves (i); besides, $V_1,V_2$ are disjoint, and thus $T = (V_1 \cup V_2, E_1 \cup E_2 \cup \{uv\})$ is an $(u,S)$-solution, which proves (ii).
\end{proof}

\subsection{The multiset case} \label{sec:multiset}


We consider the multiset formulation of the problem: we now allow some colors to be repeated but impose a maximum number of occurrences for each color. This problem can be seen as a generalization of the original $\GM$ problem.

We first introduce some notations. Given a multiset $M$ over a set $A$, and given an element $x \in A$, we denote by $n_M(x)$ the number of occurrences of $x$ in $M$.
Given two multisets $M,M'$, we denote their inclusion by $M \subseteq M'$. We denote by $|M|$ the size of $M$, where elements are counted with their multiplicities. Given two sets $A,B$, a function $f : A \rightarrow B$ and a multiset $X$ over $A$, we let $f(X)$ denote the multiset containing the elements $f(x)$ for $x \in X$, counted with multiplicities; precisely, given $y \in B$ we have $n_{f(X)}(y) = \sum_{x \in A : f(x) = y} n_X(x)$.

We now define the following two variants of \textsc{Colorful Graph Motif}, which allow for multiset motifs.\\
~

\noindent \textbf{Name:} \textsc{Multiset Graph Motif} ($\pmgm$)\\
\textbf{Input:} A graph $G = (V,E)$, a set $C$, a function $\chi : V \rightarrow C$, a multiset $M$ over $C$, an integer $k$.\\
\textbf{Solution:} A subtree $T = (V_T,E_T)$ of $G$ s.t. (i) $|V_T| = k$ and (ii) $\chi(V_T) \subseteq M$.\\
~

\noindent \textbf{Name:} \textsc{Multiset Graph Motif With Gaps} ($\pmgmg$)\\
\textbf{Input:} A graph $G = (V,E)$, a set $C$, a function $\chi : V \rightarrow C$, a multiset $M$ over $C$, integers $k,r$.\\
\textbf{Solution:} A subtree $T = (V_T,E_T)$ of $G$ s.t. (i) $|V_T| \leq r$ and (ii) there exists $S \subseteq V_T$ of size $k$ such that $\chi(S) \subseteq M$.\\

The restriction of \textsc{Multiset Graph Motif} when $|M| = k$ is called \textsc{Exact Multiset Graph Motif} ($\pxmgm$). Note that in this case we require that $T$ contains every occurrence of $M$, \textit{i.e.} $\chi(V_T) = M$.  In this way, the $\pxmgm$ problem coincides with the $\GM$ problem defined in \cite{Fellows:Fertin:Hermelin:Vialette:ICALP:2007,Betzler:Fellows:Komusiewicz:Niedermeier:CPM:2008}, while the $\pmgm$ problem is the parameterized version of the \textsc{Max Motif} problem considered in \cite{Dondi:Fertin:Vialette:CPM:2009}. The definition of the $\pmgmg$ problem encompasses the notion of insertions and deletions of~\cite{Bruckner:2009}.

Previous algorithms for these problems relied on color-coding \cite{Alon:Yuster:Zwick:1995}; these algorithms usually have an exponential space complexity, and a high time complexity. For the $\GM$ problem, \cite{Fellows:Fertin:Hermelin:Vialette:ICALP:2007} gives a randomized algorithm with an implicit $\bigo(87^k k m)$ running time, while \cite{Betzler:Fellows:Komusiewicz:Niedermeier:CPM:2008} describes a first randomized algorithm running in $\bigo(8.16^km)$, and shows a second algorithm with $\bigo(4.32^kk^2m)$ running time, using two different speed-up techniques (\cite{Bjorklund:Husfeldt:Kaski:Koivisto:2007} and \cite{Huffner:2007}). For the \textsc{Max Motif} problem,  \cite{Dondi:Fertin:Vialette:CPM:2009} presents a randomized algorithm with an implicit $\bigo((32 e^2)^k k m)$ running time. Here again, we can apply Theorem \ref{th:mld} to improve the time and space complexities.

\begin{proposition} \label{prop:mgm}
\begin{enumerate}
\item $\pmgm$ is solvable by a randomized algorithm in $\bigot(4^k k^2 m)$ time and $\bigot(k n)$ space.
\item $\pmgmg$ is solvable by a randomized algorithm in $\bigot(4^k r^2 m)$ time and $\bigot(r n)$ space.
\end{enumerate}
\end{proposition}

\begin{proof}  Point 1. We modify the circuit of Proposition \ref{prop:cgm} as follows. For each color $c \in C$ with $n_M(c) = \mu$, we introduce variables $y_{c,1}, \dots, y_{c,\mu}$, and we introduce a plus-gate $Q_c = y_{c,1} + \dots + y_{c,\mu}$. For each vertex $u \in V$, we introduce a variable $x_u$, and we define:
\begin{align*}
 P_{1,u} &= x_u Q_{\chi(u)}\\
 P_{i,u} &= \sum_{i' = 1}^{i-1} \sum_{v \in N_G(u)} P_{i',u} P_{i-i',v} \text{ if } i > 1
\end{align*}
and $P = \sum_{u \in V} P_{k,u}$. Note that we changed only the base case in the recurrence of Proposition \ref{prop:cgm}. The intuition is that the variables $x_u$ will ensure that we choose different vertices to construct the tree, and that the variables $y_{c,i}$ will ensure that a given color cannot occur more than required. The resulting instance of $\pmld$ is $I' = (\C_I,2k)$, and since $T(\C_I) = \bigo(k^2 m)$ and $S(\C_I) = \bigo(k n)$, we solve it in the claimed bounds by Theorem \ref{th:mld}. A similar induction as in Proposition \ref{prop:cgm} shows that: for every $1 \leq i \leq k$, a multilinear monomial of $P_{i,u}$ has the form $x_{v_1} y_{c_1,j_1}\dots x_{v_i} y_{c_i,j_i}$, and it is present iff there is a subtree $(V_T,E_T)$ of $G$ such that $u \in V_T$, $V_T = \{v_1,\dots,v_i\}$ and $\chi(V_T) = \lmul c_1,\dots,c_i \rmul \subseteq M$.

Point 2. We modify the construction of Point 1 by now setting $P_{1,u} = 1 + x_u Q_{\chi(u)}$ for each $u \in V$, and $P = \sum_{u \in V} \sum_{i = 1}^{r} P_{i,u}$. Informally, adding the constant $1$ to each $P_{1,u}$ permits to ignore some vertices of the subtree, allowing to only select a set $S$ of $k$ vertices such that $\chi(S) \subseteq M$. The correctness of the construction is shown by a similar induction as above. The catch here is that when considering two trees $T_1,T_2$ obtained from $P_{i',u},  P_{i-i',v}$, their selected vertices will be distinct, but they may have ``ignored'' vertices in common; we can then find a subset of $E(T_1) \cup E(T_2) \cup \{uv\}$ which forms a tree containing all selected vertices from $T_1,T_2$. 
\end{proof}

We point out that the proof of Proposition \ref{prop:mgm} can be adapted to solve the \textsc{List Colored Graph Motif} problem from \cite{Betzler:Fellows:Komusiewicz:Niedermeier:CPM:2008} in $\bigos(4^k)$ time and polynomial space. In this variant, each vertex receives a list of colors instead of only one color, but only one of these must be kept in the solution. The idea is that the node $Q_{\chi(u)}$ will be a sum over the variables corresponding to the colors of $u$. This improves upon an randomized algorithm of \cite{Betzler:Fellows:Komusiewicz:Niedermeier:CPM:2008} which runs in $\bigo(10.88^km)$ time and exponential space.

\subsection{Edge-weighted versions}\label{sec:edge-weighted}

We consider an edge-weighted variant of the problem, where the subtree is now required to have a given total weight, in addition to respecting the color constraints. This variant has been studied in \cite{Bocker:Rasche:Steijger:2009} under the name \textsc{Edge-Weighted Graph Motif}, under a slightly different definition (they indeed minimize the sum of the weight of the edges $\{u,v\}$ where only $u$ is in the solution). In our case, we define two problems, depending on whether we consider colorful or multiset motifs.\\
~

\noindent \textbf{Name:} \textsc{Weighted Colorful Graph Motif} ($\pwcgm$)\\
\textbf{Input:} A graph $G = (V,E)$, a function $\chi : V \rightarrow C$, a weight function $w : E \rightarrow \nat$, integers $k,r$.\\
\textbf{Solution:} A subtree $T = (V_T,E_T)$ of $G$ such that (i) $|V_T| = k$, (ii) $\chi$ is injective on $V_T$, (iii) $\sum_{e \in E_T} w(e) \leq r$.\\
~

\noindent \textbf{Name:} \textsc{Weighted Multiset Graph Motif} ($\pwmgm$)\\
\textbf{Input:} A graph $G = (V,E)$, a function $\chi : V \rightarrow C$, a weight function $w : E \rightarrow \nat$, a multiset $M$ over $C$, integers $k,r$.\\
\textbf{Solution:} A subtree $T = (V_T,E_T)$ of $G$ such that (i) $|V_T| = k$, (ii) $\chi(V_T) \subseteq M$, (iii) $\sum_{e \in E_T} w(e) \leq r$.\\

We observe that the $\pwmgm$ problem contains as a special case the \textsc{Min-CC} problem introduced in \cite{Dondi:Fertin:Vialette:ICTCS:2007}, which seeks a subgraph respecting the multiset motif, and having at most $r$ connected components. Indeed, we can easily reduce \textsc{Min-CC} to $\pwmgm$: given the graph $G$, we construct a complete graph $G'$ with the same vertex set, and we assign a weight 0 to edges of $G$, and a weight 1 to non-edges of $G$.

\begin{proposition} \label{prop:wgm}
\begin{enumerate}
\item $\pwcgm$ is solvable by a randomized algorithm in \sloppy $\bigot(2^k k^2 r^2 m)$ time and $\bigot(k r n)$ space.
\item $\pwmgm$ is solvable by a randomized algorithm in $\bigot(4^k k^2 r^2 m)$ time and $\bigot(k r n)$ space.
\end{enumerate}
\end{proposition}

\begin{proof} We only prove 1, since 2 relies on the same modification as in Proposition \ref{prop:mgm}. The construction of the arithmetic circuit is similar to the construction in Proposition \ref{prop:cgm}. The set of variables is $\{x_c : c \in C \}$, and we introduce nodes $P_{i,j,u}$, for $1 \leq i \leq k$ and $0 \leq j \leq r$, whose multilinear monomials will correspond to colorful subtrees having $i$ vertices including $u$, and with total weight $\leq j$. The definitions are as follows:
\begin{align*}\label{eq:wgm 1}
P_{1,j,u} &= x_{\chi(u)}\\
P_{i,j,u} &= \sum_{i'=1}^{i-1} \sum_{v \in N_G(u)} \sum_{j'=0}^{j-w(uv)} 
			    P_{i',j',u} P_{i-i',j-j'-w(uv),v} \text{ if } i > 1
\end{align*}
and $P = \sum_{u \in V} P_{k,r,u}$. The resulting instance of $\pmld$ is $I' = (\C_I,k)$, and since $T(\C_I) = \bigo(k^2 r^2 m)$ and $S(\C_I) = \bigo(k r n)$, we solve it in the claimed bounds by Theorem \ref{th:mld}. The correctness of the construction follows by showing that: given $1 \leq i \leq k, 0 \leq j \leq r$, $u \in V$, $x_{c_1} \dots x_{c_d}$ is a multilinear monomial of $P_{i,j,u}$ iff (i) $d = i$ and (ii) there exists $T = (V_T,E_T)$ colorful subtree of $G$ with $u \in V_T, \chi(V_T) = \{c_1,\dots,c_d\}$ and $\sum_{e \in E_T} w(e) \leq j$.
\end{proof}

\section{Counting vertex-colored subtrees}\label{sec:counting}

In this section, we consider the counting versions of the problems $\pxcgm$ and $\pxmgm$ introduced in Section \ref{sec:fpt}. For the former, we show that its counting version $\# \pxcgm$ is FPT; for the latter, we prove that its counting version $\# \pxmgm$ is $\swone$-hard. 

\subsection{An FPT algorithm for the colorful case}

We show that $\# \pxcgm$ is fixed-parameter tractable (Proposition \ref{prop:xcgm-counting}). We rely on a general result for $\# \pxmld$ (Proposition \ref{prop:xmld-counting}), which uses inclusion-exclusion as in \cite{K82}.

Say that a circuit $\C$ is \emph{$k$-bounded} iff $P_{\C}$ has only monomials of degree $\leq k$. Observe that given a circuit $\C$, we can efficiently transform it in a $k$-bounded circuit $\C'$ such that (i) $\C$ and $\C'$ have the same monomials of degree $k$ with the same coefficients, (ii) $|\C'| \leq (k+1)^2 |\C|$. Indeed, we can first transform $\C$ so that all internal nodes $+$ and $\times$ nodes have out-degree 2, without increasing the size; then, for each node $u$ of $\C$, we create $k+1$ nodes $u_0,\dots,u_k$, and:
\begin{itemize}
\item if $u$ is a leaf with label $v \in X$, then $u_1$ is a leaf with label $v$, and other $u_i$'s are $0$ nodes (represented by leaves labeled by $+$);
\item if $u$ is a leaf with label $l \in \{+,\times\}$, then $u_0$ is a leaf with label $l$, and other $u_i$'s are $0$ nodes;
\item if $u = v+w$, then for every $i$, $u_i = v_i + w_i$;
\item if $u = v \times w$, then for every $i$, $u_i = \sum_{j = 0}^{i} v_j w_{i-j}$.
\end{itemize}
Let $\C'$ be the resulting circuit; if $r$ is the root of $\C$, then $r_k$ becomes the root of $\C'$. It is easily checked that $\C'$ has the same monomials of degree $k$ as the original circuit $\C$. Besides, $|\C'| \leq (k+1)^2 |\C|$ since for each node $u$ of $\C$, we have introduced $k+1$ nodes each of out-degree $\leq k+1$ in $\C'$.


The following result shows that we can efficiently count solutions for $k$-bounded circuits with $k$ variables (and thus for general circuits, with an extra $\bigo(k^2)$ factor in the complexity).

\begin{proposition} \label{prop:xmld-counting} $\# \pxmld$ for $k$-bounded circuits is solvable in $\bigo(2^k T(\C))$ time and $\bigo(S(\C))$ space.
\end{proposition}

\begin{proof} Let $\C$ be the input circuit on a set $X$ of $k$ variables. For a monomial $m$ let $Var(m)$ denote its set of variables. Given $S \subseteq X$, let $N_S$, resp. $N'_S$, be the number of monomials $m$ of $P_{\C}$ such that $Var(m) = S$, resp. $Var(m) \subseteq S$. Observe that for every $S \subseteq X$, we have $N'_S = \sum_{T \subseteq S} N_T$. Therefore, by M\"obius inversion it holds that for every $S \subseteq X$, $N_S = \sum_{T \subseteq S} (-1)^{|S \backslash T|} N'_T$.

Since $\C$ is $k$-bounded, $N_X$ is the number of multilinear monomials of $P_{\C}$ having degree $k$. Now, each value $N'_S$ can be computed by evaluating $\C$ under the mapping $\phi : X \rightarrow \Z$ defined by $\phi(v) = 1$ if $v \in S$, $\phi(v) = 0$ if $v \notin S$. This mapping gives the right number of monomials $m$ such that $Var(m) \subseteq S$. Indeed, if all the variables of a monomial $m$ are in $S$, $m$ is evaluated to 1. Otherwise, if one variable of $m$ is not in $S$, $m$ is evaluated to 0.
Therefore, $N'_S$ can be computed in $\bigo(T(\C))$ time and $\bigo(S(\C))$ space.
By the M\"obius inversion formula, we can thus compute the desired value $N_X$ in $\bigo(2^k T(\C))$ time and $\bigo(S(\C))$ space. 

\end{proof}

It is worth mentioning that Proposition \ref{prop:xmld-counting} generalizes several counting algorithms based on inclusion-exclusion, such as the well-known algorithm for $\# \textsc{Hamiltonian Path}$ of \cite{K82}, as well as results of \cite{N09}. Indeed, the problems considered in these articles can be reduced to counting multilinear monomials of degree $n$ for circuits with $n$ variables (where $n$ is usually the number of vertices of the graph), which leads to algorithms running in $\bigos(2^n)$ time and polynomial space.

Let us now turn to applying Proposition \ref{prop:xmld-counting} to the $\# \pxcgm$ problem. Recall that we defined in Proposition \ref{prop:cgm} a circuit $\C_I$ for the general $\pcgm$ problem; we will have to modify it slightly for the purpose of counting solutions.

\begin{proposition} \label{prop:xcgm-counting} $\# \pxcgm$ is solvable in $\bigo(2^k k^3 m)$ time and $\bigo(k^2 n)$ space.
\end{proposition}

\begin{proof} Let $I$ be an instance of $\pxcgm$. A \emph{rooted solution} for $I$ is a pair $(T,u)$ where $T$ is a solution of $\pxcgm$ on $I$ and $u$ is a vertex of $T$ (which should be seen as the root of the tree). The solutions of $\pxcgm$ on $I$ are also called \emph{unrooted solutions}. Let $N_r(I)$ and $N_u(I)$ be the number of rooted, resp. unrooted, solutions for $I$. We will show how to compute $N_r(I)$ in the claimed time and space bounds; since $N_u(I) = \frac{N_r(I)}{k}$, the result will follow.

To compute $N_r$, observe first that we cannot apply Proposition \ref{prop:xmld-counting} to the circuit $\C_I$ of Proposition \ref{prop:cgm}. Indeed, the circuit $\C_I$ counts the ordered subtrees, and not the unordered ones. Therefore, we need to modify the circuit in the following way: at each vertex $v$ of $V_T$, we examine its children by increasing color. This leads us to define the following circuit $\C'_I$: suppose w.l.o.g. that $C = \{1,\dots,k\}$, introduce nodes $P_{i,j,u}$ for each $1 \leq i \leq k, 1 \leq j \leq k+1, u \in V$, variables $x_i$ for each $1 \leq i \leq k$, and define:
\begin{align*}
P_{1,j,u} &= x_{\chi(u)}\\
P_{i,j,u} &= 0 \text{ if } i \geq 2, j = k+1\\
P_{i,j,u} &= P_{i,j+1,u} + \sum_{i' = 1}^{i-1} \sum_{v \in N_G(u) : \chi(v) = j} P_{i',j+1,u} P_{i-i',1,v} \text{ if } i \geq 2, 1 \leq j \leq k
\end{align*}
Let us also introduce a root node $P = \sum_{u \in V} P_{k,1,u}$. Given $1 \leq i \leq k$, $1 \leq j \leq k+1$ and $u \in V$, let $\S_{i,j,u}$ denote the set of pairs $(T,u)$ where (i) $T$ is a distinctly colored subtree of $I$ containing $u$ and having $i$ vertices, (ii) the neighbors of $u$ in $T$ have colors $\geq j$. It can be shown by induction on $i$ that: there is a bijection between $\S_{i,j,u}$ and the multilinear monomials of $P_{i,j,u}$. Therefore, the number of multilinear monomials of $P$ is equal to $N_r$; since $T(\C'_I) = \bigo(k^3 m), S(\C'_I) = \bigo(k^2 n)$ and since $\C'_I$ is $k$-bounded, it follows by Proposition \ref{prop:xmld-counting} that $N_r$ can be computed in $\bigo(2^k k^3 m)$ time and $\bigo(k^2 n)$ space.

\end{proof}

Observe that Lemma 2.1 of~\cite{Arvind2002} already gives a deterministic FPT algorithm for $\# \pxcgm$. The time and space complexities of Proposition~\ref{prop:xcgm-counting} are however lower.

\subsection{Hardness of the multiset case}

In this subsection, we show that $\# \pxmgm$ is $\swone$-hard. For convenience, we first restate the problem in terms of \emph{vertex-distinct embedded subtrees}.

Let $G = (V,E)$ and $H = (V',E')$ be two multigraphs. An \emph{homomorphism} of $G$ into $H$ is a pair $\phi = (\phi_V,\phi_E)$ where $\phi_V : V \rightarrow V'$ and $\phi_E : E \rightarrow E'$, such that if $e \in E$ has endpoints $x,y$ then $\phi_E(e)$ has endpoints $\phi_V(x),\phi_V(y)$. An \emph{embedded subtree} of $G$ is denoted by $\T = (T,\phi_V,\phi_E)$ where $T = (V_T,E_T)$ is a tree, and $(\phi_V,\phi_E)$ is an homomorphism from $T$ into $G$. We say that $\T$ is a \emph{vertex-distinct} embedded subtree of $G$ (a ``vdst'' of $G$) if $\phi_V$ is injective. We say $\T$ is an \emph{edge-distinct} embedded subtree of $G$ (an ``edst'' of $G$) iff $\phi_E$ is injective. We restate $\pxmgm$ as follows.\\
~

\noindent \textbf{Name:} \textsc{Exact Multiset Graph Motif} ($\pxmgm$)\\
\textbf{Input:} A graph $G = (V,E)$, an integer $k$, a set $C$, a function $\chi : V \rightarrow C$, a multiset $M$ over $C$ s.t. $|M| = k$.\\
\textbf{Solution:} A vdst $(T,\phi_V,\phi_E)$ of $G$ s.t. $\chi \circ \phi_V(V_T) = M$.
~\\

We first show the hardness of two intermediate problems (Lemma \ref{lem:mest-counting}). Before defining these problems, we need the following notions. 
Consider a multigraph $G = (V,E)$. Consider a partition $\P$ of $V$ into $V_1,\dots,V_k$, and a tuple $t \in [r]^k$. A \emph{$(\P,t)$-mapping} from a set $A$ is an injection $\psi : A \rightarrow V \times [r]$ such that for every $x \in A$, if $\psi(x) = (v,i)$ with $v \in V_j$, then $1 \leq i \leq t_j$. From $\psi$, we define its \emph{reduction} as the function $\psi^r : A \rightarrow V$ defined by $\psi^r(x) = v$ whenever $\psi(x) = (v,i)$. We also define a tuple $T(\psi) = (n_1,\dots,n_k) \in [r]^k$ such that for each $i \in [k]$, $n_i = \max_{v \in V_i} | \{ x \in A : \psi^r(x) = v \} |$.

Given two tuples $t,t' \in [r]^k$, denote $t \leq t'$ iff $t_i \leq t'_i$ for each $i \in [k]$. Note that for a $(\P,t)$-mapping $\psi$, we always have $T(\psi) \leq t$ since $\psi$ is injective. We say that a \emph{$(\P,t)$-labeled edst for $G$} is a tuple $(T,\psi_V,\psi_E)$ where (i) $T = (V_T,E_T)$ is a tree, (ii) $\psi_V$ is a $(\P,t)$-mapping from $V_T$, (iii) $(T,\psi_V^r,\psi_E)$ is an edst of $G$. Our intermediate problems are defined as follows.\\
~

\noindent \textbf{Name:} \textsc{Multicolored Embedded Subtree-1} ($\pmest-1$)\\
\textbf{Input:} Integers $k,r$, a $k$-partite multigraph $G$ with partition $\P$, a tuple $t \in [r]^k$.\\
\textbf{Solution:} A $(\P,t)$-labeled edst $(T,\psi_V,\psi_E)$ for $G$ s.t. $|V_T| = r$ and $T(\psi_V) = t$.\\
~

The $\pmest-2$ problem is defined similarly, except that we do not require that $T(\psi_V) = t$ (and thus we only have $T(\psi_V) \leq t$). While we will only need $\# \pmest-2$ in our reduction for $\# \pxmgm$, we first show the hardness of $\# \pmest-1$, then reduce it to $\# \pmest-2$.

\begin{lemma} \label{lem:mest-counting} $\# \pmest-1$ and $\# \pmest-2$ are $\swone$-hard for parameter $(k,r)$.
\end{lemma}

\begin{proof}

We first reduce $\# \textsc{Multicolored Clique}$ to $\# \pmest-1$. Our source problem $\# \textsc{Multicolored Clique}$ is the counting version of \textsc{Multicolored Clique}, which is easily seen to be $\swone$-hard (from the $\swone$-hardness of $\#\textsc{Clique}$ \cite{Flum:Grohe:2006}). Let $I = (G,k)$ be an instance of the problem, where $G = (V,E)$ has a partition $\P$ into classes $V_1,\dots,V_k$. Our target instance is $I' = (k,r,H,t)$ with $r = k^2-k+1$ and $t = (k,k-1,\dots,k-1)$. The graph $H$ is obtained from $G$ by splitting every edge $e$ in two parallel edges; then $H$ is a $k$-partite multigraph with partition $\P$. Let $\S_I$, $\S_{I'}$ be the solution sets of $I$ and $I'$ respectively. Let $\K_k$ be the multigraph with $k$ vertices $1,\dots,k$, and with two parallel edges between distinct vertices; its partition is $\P_k$ consisting of the sets $\{1\},\dots,\{k\}$. Let $\U_k$ denote the set of $(\P_k,t)$-labeled edsts $(\T,\psi_V,\psi_E)$ for $\K_{k}$ such that $T(\psi_V) = t$. Observe that $\U_k \neq \emptyset$: since every vertex of $\K_k$ has degree $2(k-1)$, it follows that $\K_k$ has an Eulerian path starting at 1, which visits $k$ times the vertex 1, and each other vertex $k-1$ times. We claim that $|\S_{I'}| = |\U_{k}| |\S_I|$, which will prove the correctness of the reduction. To this aim, we will describe a bijection $\Phi : \S_{I} \times \U_k \rightarrow \S_{I'}$.

Consider a pair $P = (C,\T) \in \S_{I} \times \U_k$ with $\T = (T,\psi_V,\psi_E)$ and $C = \{v_1,\dots,v_k\}$ multicolored clique of $G$ (with $v_i \in V_i$). Let $\phi = (\phi_V,\phi_E)$ be the homomorphism of $\K_k$ into $H$ which maps $i$ to $v_i$, and the parallel edges accordingly. We then define $\T' = \Phi(P)$ by $\T' = (T,\psi'_V, \psi'_E)$, where (i) $\psi'_V$ is defined so that if $\psi_V(u) = (v,i)$ and if $\phi_V(v) = w$ then $\psi'_V(u) = (w,i)$, (ii) $\psi'_E = \psi_E \circ \phi_E$. We verify that $\T' \in \S_{I}$: indeed, it is a $(\P,t)$-labeled edst of $G$ and $T(\psi'_V) = t$ (since we have composed with injective functions $\phi_V$, $\phi_E$). To prove that $\Phi$ is a bijection, we define the inverse function $\Psi : \S_{I} \rightarrow \S_{I'} \times \U_k$ as follows. Consider $\T' = (T,\psi'_V,\psi'_E)$ $(\P,t)$-labeled edst of $G$, with $T(\psi'_V) = t$. This equality yields vertices $v_1 \in V_1,\dots,v_k \in V_k$ such that $|(\psi^r_V)^{-1}(v_i)| = t_i$. Let $C = \{v_1,\dots,v_k\}$, then $C$ is a multicolored clique of $G$: indeed, $H[C]$ has at most $k^2-k$ edges, and since $\psi'_E$ is injective it must have exactly $k^2-k$ edges, implying that $G[C]$ is a complete graph. We can then define $(\psi_V,\psi_E)$ from $(\psi'_V,\psi'_E)$ by ``projecting'' $v_i$ on $i$, and the parallel edges accordingly (for instance, if $\psi'_V(u) = (v_i,j)$ then $\psi_V(u) = (i,j)$). We finally define $P = \Psi(\T')$ by $P = (C,\T)$ where $\T = (T,\psi_V,\psi_E)$. It is easy to see that $P \in \S_{I'} \times \U_k$, and that $\Phi$ and $\Psi$ are inverse of each other.

We now give a Turing-reduction of $\# \pmest-1$ to $\# \pmest-2$. Given a tuple $t \in [r]^k$, we define the instance $I_t = (k,r,G,t)$, and we let $\S_t, \S'_t$ be its solution sets for $\# \pmest-1, \# \pmest-2$ respectively. Let $N_t = |\S_t|$ and $N'_t = |\S'_t|$. We have for every $t \in [r]^k$: $N'_t = \sum_{t' \leq t} N_{t'}$, which yields by M\"obius inversion that for every $t \in [r]^k$: $N_t = \sum_{t' \leq t} \mu(t,t') N'_{t'}$ \footnote{where $\mu(t,t')$ is 0 if there exists $i \in [k]$ s.t. $t_i - t'_i > 1$, and is otherwise equal to $(-1)^r$ where $r$ is the number of $i \in [k]$ s.t. $t_i - t'_i = 1$.}. Therefore, we can compute a value $N_t$ using $\bigo(2^k)$ oracle calls for $\# \pmest-2$, thereby solving $\# \pmest-1$. 

\end{proof}

\begin{proposition} \label{prop:xmgm-counting} $\# \pxmgm$ is $\swone$-hard for parameter $k$.
\end{proposition}

\begin{proof}
We reduce from $\# \pmest-2$, and conclude using Lemma \ref{lem:mest-counting}. Let $I = (k,r,G,t)$ be an instance of $\# \pmest-2$, where $G = (V,E)$ is a multigraph, and let $\S_I$ be its set of solutions. From $G$, we construct a graph $H$ as follows: (i) we subdivide each edge $e \in E$, creating a new vertex $a[e]$, (ii) we substitute each vertex $v \in V_i$ by an independent set formed by $t_i$ vertices $b[v,1],\dots,b[v,t_i]$. We let $A$ be the set of vertices $a[e]$ and $B$ the set of vertices $b[v,i]$, we therefore have a bipartite graph $H = (A \cup B,F)$. We let $I' = (H,2r-1,C,\chi,M)$, where $C = \{1,2\}$, $\chi$ maps $A$ to $1$ and $B$ to $2$, and $M$ consists of $r-1$ occurrences of 1 and $r$ occurrences of $2$.

Then $I'$ is our resulting instance of $\# \pxmgm$, and we let $\S_{I'}$ be its set of solutions. Notice that by definition of $\chi$ and $M$, $\S_{I'}$ is the set of vdst $(T,\phi_V,\phi_E)$ of $H$ containing $r-1$ vertices mapped to $A$ and $r$ vertices mapped to $B$. We now show that we have a parsimonious reduction, by describing a bijection $\Phi : \S_I \rightarrow \S_{I'}$. Consider $\T = (T,\psi_V,\psi_E)$ in $\S_I$; we define $\Phi(\T) = (T',\phi_V,\phi_E)$ as follows:
\begin{itemize}
\item For each edge $e = uv \in E(T)$, we have $f_e := \psi_E(e) \in E(G)$: we then subdivide $e$, creating a new vertex $x_e$. Let $T'$ be the resulting tree;
\item For each vertex $x_e$, we define $\phi_V(x_e) = a[f_e]$. For each other vertex $u$ of $T'$, we have  $u \in V(T)$, let $(v,i) = \psi_V(u)$; we then set $\phi_V(u) = b[v,i]$ (this is possible since if $v \in V_j$ then $1 \leq i \leq t_j$, by definition of $\psi_V$).
\end{itemize}
From $\phi_V$, we then define $\phi_E$ in a natural way.
Then $\T' = \Phi(\T)$ is indeed in $\S_{I'}$: (i) $\T'$ is a vertex-distinct subtree of $H$ (by definition of $\phi_V$ and since $\T$ was edge-distinct, the values $\phi_V(x_e)$ are distinct; by injectivity of $\psi_V$, the other values $\phi_V(u)$ are distinct); (ii) it has $r-1$ vertices mapped to $A$ and $r$ vertices mapped to $B$. To prove that $\Phi$ is a bijection, we describe the inverse correspondence $\Psi : \S_{I'} \rightarrow \S_I$. Consider $\T' = (T',\phi_V,\phi_E)$ in $\S_{I'}$; we define $\Psi(\T') = (T,\psi_V,\psi_E)$ as follows. Let $A',B'$ be the vertices of $T'$ mapped to $A,B$ respectively. Let $i$ be the number of nodes of $A'$ which are leaves: since the nodes of $A'$ have degree 1 or 2 in $T'$ depending on whether they are leaves or internal nodes, we then have $|E(T')| \leq i + 2 (r-1-i) = 2r-i-2$; since $|E(T')| = 2r-2$, we must have $i = 0$. It follows that all leaves of $T'$ belong to $B'$; from $T'$, by contracting each vertex of $A'$ in $T'$ we obtain a tree $T$ with $r$ vertices. We then define $\psi_V,\psi_E$ as follows: (i) given $u \in B'$, if $\phi_V(u) = b[v,j]$, then $\psi_V(u) = (v,j)$; (ii) given $e = uv \in E(T)$, there corresponds two edges $ux,vx \in E(T')$ with $x \in A'$, and we thus have $\phi_V(x) = a[f]$, from which we define $\psi_E(e) = f$. It is easily seen that the resulting $\T = \Psi(\T')$ is in $\S_{I}$, and that the operations $\Phi$ and $\Psi$ are inverse of each other. 
\end{proof}

\section{Practical evaluation}

We implemented in Java the algorithm of Proposition \ref{prop:xcgm-counting} to compare the multilinear detection framework with known techniques used to solve \GM. To do so, our tests consist in retrieving motifs (protein complexes) of six different species in three large different Protein-Protein Interaction networks and in comparing the running time of our algorithm with \prog{GraMoFoNe} \cite{Blin:Sikora:Vialette:2010} and \prog{Torque} \cite{Bruckner:2009}. Note that our implemented algorithm of Proposition~\ref{prop:xcgm-counting} counts the occurrences of a motif, while \prog{GraMoFoNe} and \prog{Torque} do not perform counting. Therefore, as an alternative to our counting algorithm we also used the circuit of Proposition~\ref{prop:cgm} with Proposition~\ref{prop:xmld-counting}, which is faster and sufficient for a decision purpose. This allows us to perform a more fair comparison between decision algorithms, and also to compare the running times of the decision and counting algorithms.

\subsection{Data Acquisition}

The networks (of Saccharomyces cerevisiae (yeast), Drosophila melanogaster (fly) and \sloppy Homo sapiens) are those collected by the authors of \prog{Torque} \cite{Bruckner:2009} from recent papers and online databases. Their sizes are between 5000 and 8000 proteins, and between 20.000 and 40.000 interactions. The motifs are proteins complexes of Saccharomyces cerevisiae, Drosophila melanogaster, Homo sapiens, Mus musculus (mouse), Bos taurus (Bovine) and Rattus norvegicus (rat) also collected by \prog{Torque} authors from online databases. The FASTA files are those collected from online databases by the authors of \prog{GraMoFoNe} \cite{Blin:Sikora:Vialette:2010}.

\subsection{Settings}

We tried to use the same settings and parameters in the three algorithms. Since Proposition \ref{prop:xcgm-counting} allows only exact matches, we set to 0 the number of possible insertions and deletions in \prog{GraMoFoNe} and \prog{Torque}. The timeout limit for the three algorithms was set to 500 seconds.


\subsection{Experiments}

All algorithms were executed on a standard desktop PC (3GHz and 2Gb RAM). \prog{GraMoFoNe} is based on a pseudo boolean solver, while \prog{Torque} is based on a dynamic programming algorithm (\prog{Torque} also uses Integer Linear Programming but we do not use it during our tests). 

The input is a colorful motif $C$ (the motif is completely defined by adding a color for each different protein present in the protein complex) and a vertex-colored network $G$. A node of $G$ is colored by a color $c \in C$ if the protein represented by this node is homologous to the protein represented by $c$ (according to a BLASTp analysis). 

Before running the algorithm, one can remark that since insertions are not allowed, we can safely remove each non-colored node of the network. This step greatly prunes the network since in practice, only 5\% of the nodes are colored (according to \cite{Bruckner:2009}).

We launched the three algorithms for each feasible complex of each species, with the (pruned) network of each species (except the one of the complex). As in \prog{Torque}, a complex is called feasible if (i) the size of the complex is between 4 and 25 (both included), and (ii) there is a connected component containing all colors of the complex (since no deletions are allowed). We then computed for each feasible complex the running time of each algorithm to find a solution, or to conclude that there is no solution. For information, 70\% of the feasible complexes have size 4 or 5. We did not count the running time when the algorithm reached its timeout. 

Our algorithm did not support multiple colors for each network node, and no insertions and deletions were allowed. Therefore, we did not compare the solutions found by each algorithm for each complex, since biological data are too noisy for such results to be realistic.

\subsection{Comparison with related works}

\begin{figure}
\begin{center}
\begin{tabular}{c |c c c c}
  & Prop.~\ref{prop:cgm} & Prop.~\ref{prop:xcgm-counting} & \prog{GraMoFoNe} & \prog{Torque} \\ 
\hline 4 & $<0.1$ 	& $0.1$ 	& $<0.1$ 	& $1.6$ \\ 
\hline 5 & $<0.1$	& $<0.1$	& $<0.1$	& $2$ 	\\ 
\hline 6 & $<0.1$ 	& $<0.1$ 	& $<0.1$ 	& $2.4$ \\ 
\hline 7 & $<0.1$ 	& $<0.1$ 	& $<0.1$ 	& $2.5$ \\ 
\hline 8 & $<0.1$ 	& $0.1$ 	& $<0.1$ 	& $3.6$ \\ 
\hline 9 & $0.1$ 	& $1.5$ 	& $0.2$ 	& $3.4$ \\ 
\end{tabular} 
\end{center}
\caption{Comparison in seconds between the two versions of our algorithm with \prog{Torque} and \prog{GraMoFoNe}. For each different size, the computed value is the average running time of the algorithm, executed with each network and each feasible protein complex of this size from each species.}\label{fig:comparison}
\end{figure}

The average running times of the three algorithms launched over all feasible complexes can be found on Figure~\ref{fig:comparison}. We show results for complexes of size up to 9 only since there are very few feasible complexes of size greater than 9 when one did not allow insertions and deletions. One can note that our algorithm must complete the same number of operations independently of the presence of a solution, \textit{i.e.} we have to evaluate the circuit for all subsets of $\{1,\dots,k\}$ due to Proposition~\ref{prop:xmld-counting}. 

Our results demonstrate that the multilinear detection framework lends itself to implementation, and is competitive with other techniques. Still, one has to be careful in interpreting these results. On the one hand, the tools are implemented in different languages: Java in the case of \prog{GraMoFoNe} and of our algorithm, Python in the case of \prog{Torque}. On the other hand, we compared the algorithms in the situation where no insertions or deletions are allowed, which is unrealistic from a biological viewpoint. In order to allow a more comprehensive comparison of the programs, and to have an algorithm applicable to real biological purpose, some work remains to be done. In particular, it is desirable to add support for multiset motifs and insertions-deletions. It may be done by implementing Proposition~\ref{prop:mgm}, which implies to implement Koutis-Williams algorithm of Theorem~\ref{th:mld}.


\section{Conclusion}

In this paper, we have obtained improved FPT algorithms for several variants of the $\GM$ problem. Reducing to the \textsc{Multilinear Detection} problem resulted in faster running times and a polynomial space complexity. We have also considered the counting versions of these problems, establishing a complexity dichotomy between the colorful and multiset cases. 
Our results demonstrate that the algebraic framework of \cite{Koutis:Williams:ICALP:2009} has potential applications to computational biology, since our implemented algorithms based on Proposition~\ref{prop:xmld-counting} achieve comparable performance with existing software.

We conclude with some open questions. A first question concerns our results of Section \ref{sec:multiset} for multiset motifs: is it possible to further reduce the $\bigos(4^k)$ running times? Another question relates to the edge-weighted problems considered in Section \ref{sec:edge-weighted}: our algorithms are only pseudopolynomial in the maximum weight $r$, can this dependence in $r$ be improved? Finally, is approximate counting possible for the $\# \pxmgm$ problem? We believe that some of these questions may be solved through an extension of the algebraic framework of Koutis and Williams.

\section{Acknowledgement}
The authors acknowledge partial funding from DFG PABI BO1910/9-1 and ANR project BIRDS JCJC SIMI 2-2010, and also would like to thanks Sharon Bruckner, Khanh-Lam Mai and the anonymous reviewers for valuable comments and remarks. The final publication is available at \url{http://www.springerlink.com/content/u84x683503577735/}.

\bibliographystyle{spmpsci}

\end{document}